\documentclass{cccg21}
\usepackage{graphicx}[color]
\usepackage{color}
\usepackage{amssymb,amsmath}
\usepackage{tcolorbox}
\tcbset{boxrule=0.25mm,left=0.4em,right=0.4em,top=0.4em,bottom=0.4em,enlarge left by=0.5em,width=\linewidth-1em}

\graphicspath{{figures/}}

\usepackage[utf8]{inputenc} 
\usepackage{CJKutf8} 

\usepackage{hyperref}
\usepackage{cleveref}
\usepackage{subcaption}
\usepackage[percentage]{overpic}

\def\defn#1{\textbf{\textit{\boldmath #1}}}
\let\emph=\defn

\newcounter{section-preserve}
\newcounter{theorem-preserve}
\newcommand{\blank}[1]{}
\newtoks\magicAppendix
\magicAppendix={}
\newtoks\magictoks
\newif\iflater
\laterfalse
\long\def\later#1{\magictoks={#1}%
	\edef\magictodo{\noexpand\magicAppendix={\the\magicAppendix 
			\the\magictoks%
	}}
	\magictodo}
\long\def\both#1{\magictoks={#1}%
	\edef\magictodo{\noexpand\magicAppendix={\the\magicAppendix 
			\noexpand\setcounter{theorem-preserve}{\noexpand\arabic{theorem}}%
			\noexpand\setcounter{theorem}{\arabic{theorem}}%
			\noexpand\setcounter{section-preserve}{\noexpand\arabic{section}}%
			\noexpand\setcounter{section}{\arabic{section}}%
			\noexpand\let\noexpand\oldsection=\noexpand\thesection
			\noexpand\def\noexpand\thesection{\thesection}
			\noexpand\let\noexpand\oldlabel=\noexpand\label
			\noexpand\let\noexpand\label=\noexpand\blank
			\the\magictoks%
			\noexpand\setcounter{theorem}{\noexpand\arabic{theorem-preserve}}%
			\noexpand\setcounter{section}{\noexpand\arabic{section-preserve}}%
			\noexpand\let\noexpand\thesection=\noexpand\oldsection
			\noexpand\let\noexpand\label=\noexpand\oldlabel
	}}
	\magictodo
	\the\magictoks}
\def\magicappendix{\latertrue \the\magicAppendix}

{\makeatletter
	\gdef\xxxmark{%
		\expandafter\ifx\csname @mpargs\endcsname\relax 
		\expandafter\ifx\csname @captype\endcsname\relax 
		\marginpar{xxx}
		\else
		xxx 
		\fi
		\else
		xxx 
		\fi}
	\gdef\xxx{\@ifnextchar[\xxx@lab\xxx@nolab}
	\long\gdef\xxx@lab[#1]#2{{\bf [\xxxmark #2 ---{\sc #1}]}}
	\long\gdef\xxx@nolab#1{{\bf [\xxxmark #1]}}
	\long\gdef\xxx@lab[#1]#2{}\long\gdef\xxx@nolab#1{}%
}

\let\realbfseries=\bfseries
\def\bfseries{\realbfseries\boldmath}
\makeatletter
\let\real@titlestyle=\@titlestyle
\def\@titlestyle{\real@titlestyle\boldmath}
\makeatother

\author{Erik D. Demaine\thanks{Massachusetts Institute of Technology, USA, \protect\url{edemaine@mit.edu}}
	\and
	Jayson Lynch\thanks{University of Waterloo, Canada, \protect\url{jayson.lynch@uwaterloo.ca}}
\and
Mikhail Rudoy\thanks{LeapYear Technologies, \protect\url{mrudoy@gmail.com}}
\and
Yushi Uno\thanks{Osaka Prefecture University, \protect\url{uno@cs.osakafu-u.ac.jp}}
}

\title{Yin-Yang Puzzles are NP-complete}

\begin{document}

\maketitle

\begin{abstract}
  We prove NP-completeness of Yin-Yang / Shiromaru-Kuromaru
  pencil-and-paper puzzles.  Viewed as a graph partitioning problem,
  we prove NP-completeness of partitioning a rectangular grid graph
  into two induced trees (normal Yin-Yang),
  or into two induced connected subgraphs (Yin-Yang without $2 \times 2$ rule),
  subject to some vertices being pre-assigned to a specific tree/subgraph.
\end{abstract}

\section{Introduction}

The Yin-Yang puzzle is an over-25-year-old type of pencil-and-paper logic puzzle,
the genre that includes Sudoku and many other puzzles made famous by e.g.\
Japanese publisher Nikoli.%
\footnote{However, Yin-Yang is not a Nikoli puzzle.}
\figurename~\ref{fig:spiral} shows a simple example of a puzzle and its
solution.
In general, a Yin-Yang puzzle consists of a rectangular $m \times n$ grid
of unit squares, called \emph{cells},
where each cell either has a black circle, has a white circle, or is empty.
The goal of the puzzle is to fill each empty cell with either a black circle
or a white circle to satisfy the following two constraints:
\begin{itemize}
\item \emph{Connectivity constraint}:
  For each color (black and white), the circles of that color
  form a single connected group of cells,
  where connectivity is according to four-way orthogonal adjacency.
\item \emph{$2 \times 2$ constraint}:
  No $2 \times 2$ square contains four circles of the same color.
\end{itemize}
In this paper, we prove NP-completeness of deciding whether a Yin-Yang
puzzle has a solution, with or without the $2 \times 2$ constraint.

%
%

\subsection{Graph Partitioning}

We can view Yin-Yang puzzles as a type of graph partitioning problem,
by taking the dual graph with a vertex for each unit-square cell and
edges between orthogonally adjacent vertices/cells.
The result is a rectangular $m \times n$ grid graph, with some vertices
precolored black or white.
The goal is to complete a black/white coloring of the vertices
subject to dual versions of the constraints.
The connectivity constraint above
constrains each color class to induce a connected subgraph,
so with this constraint alone, the problem is to partition the graph's vertices
into two connected induced subgraphs.
We thus also prove NP-completeness of this graph partitioning problem:

\begin{tcolorbox}
\emph{Grid Graph Connected Partition Completion}:
Given an $m \times n$ grid graph $G = (V,E)$ and
given a partition of $V$ into $A$, $B$, and $U$, 
is there a partition of $V$ into $A'$ and $B'$ such that 
$A \subseteq A'$, $B \subseteq B'$, and
$G[A']$ and $G[B']$ are connected?
\end{tcolorbox}

\begin{figure}
  \centering
  \includegraphics{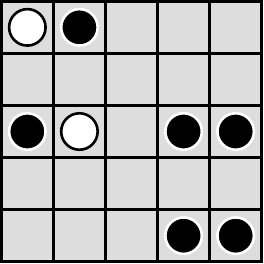}\hfil
  \includegraphics{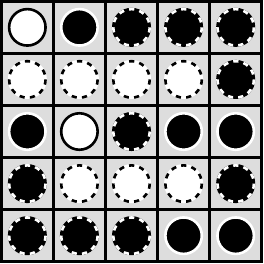}
  \caption{A simple Yin-Yang puzzle (left) and its unique solution (right).
    Circles added in the solution have dotted boundaries.}
  \label{fig:spiral}
\end{figure}

The $2 \times 2$ constraint forbids induced 4-cycles in each color class,
which is the \emph{thin} constraint of \cite{arkin2009not,demaine2017thin}.
Any larger-than-4 induced cycle in a color class must enclose a vertex of the
opposite color, so by the connectivity constraint,
only one color class can have such a cycle
and it can have only one such cycle, which (if it exists) must be exactly
the boundary vertices of the $m \times n$ grid graph.
If we exclude the possibility of a single outer cycle (e.g., by restricting to
instances that precolor at least one boundary vertex of each color), then the 
problem (with both constraints) is to partition the vertices
of a rectangular $m \times n$ grid graph into two induced (connected) trees.
We also prove NP-completeness of this graph partitioning problem:

\begin{tcolorbox}
\emph{Grid Graph Tree Partition Completion}:
Given an $m \times n$ grid graph $G = (V,E)$ and
given a partition of $V$ into $A$, $B$, and $U$, 
is there a partition of $V$ into $A'$ and $B'$ such that 
$A \subseteq A'$, $B \subseteq B'$, and
$G[A']$ and $G[B']$ are trees?
\end{tcolorbox}

\subsection{History}

The origin of Yin-Yang is not clear. 
An early example of this puzzle is from 1994 in the (discontinued)
Japanese puzzle magazine \textit{Puzzler} \cite{yin-yang-original}
(or its original form in 1993).
This reference calls the puzzle by the domestic Japanese name
``\begin{CJK}{UTF8}{ipxm}白丸黒丸\end{CJK}''
(``Shiromaru-Kuromaru'', which means ``white circle / black circle''). 
It seems that the puzzle or slight variations have been invented independently 
many times under different names. 
Most recently, it is often referred to as the ``Yin-Yang'' puzzle, 
and we could find the puzzle introduced in some puzzle books, magazines, and websites, e.g., \cite{yin-yang-example}. 

The computational complexity of puzzles has seen significant study, partly for the recreational element but also because many puzzles have direct connection to important problems such as geometric packing/partitioning or path/tree drawing under constraints.
Surveys have been written on the topic of the computational complexity of games and puzzles \cite{Demaine-Hearn-2009-survey-both,GPC,kendall2008survey}.
Many pencil-and-paper puzzles have been shown to be NP-complete, including:
Bag / Corral \cite{Friedman-2002-corral},
Country Road \cite{YajilinandCountryRoad},
Dosun-Fuwari \cite{iwamoto2018dosun},
Fillomino \cite{Yato-2003},
Fillmat \cite{uejima2015fillmat}
Hashiwokakero \cite{Hashiwokakero},
Heyawake \cite{Holzer-Ruepp-2007},
Hiroimono / Goishi Hiroi \cite{Andersson-2007},
Hitori \cite[Section~9.2]{GPC},
Juosan \cite{iwamoto2020polynomial}
Kakuro / Cross Sum \cite{Yato-Seta-2003},
Kurodoko \cite{Kurodoko},
Kurotto \cite{iwamoto2020polynomial},
Light Up / Akari \cite{McPhail-2005-LightUp},
LITS \cite{McPhail-2007-LITS},
Masyu / Pearl \cite{Friedman-2002-pearl},
Nonogram / Paint By Numbers \cite{Ueda-Nagao-1996},
Numberlink \cite{NumberlinkNP},
Nurikabe \cite{McPhail-2003,Holzer-Klein-Kutrib-2004},
Pencils \cite{is2018paper},
Shakashaka \cite{Shakashaka,NumberlessShakashaka_CCCG2015},
Slitherlink \cite{Yato-Seta-2003,Yato-2003,Witness_FUN2018},
Spiral Galaxies / Tentai Show \cite{SpiralGalaxies},
Sto-Stone \cite{allen2018sto},
Sudoku \cite{Yato-Seta-2003,Yato-2003},
Tatamibari \cite{adler2020tatamibari},
Usowan \cite{iwamoto2018computational},
Yajilin \cite{YajilinandCountryRoad},
and
Yosenabe \cite{Yosenabe}.

Graph partitioning problems involve splitting the vertices of a graph into
subsets based on various criteria. Common criteria include connectedness
of the subgraphs, balancing the number or weight of vertices or edges in the
partitions, and minimizing or maximizing the edge cut between partitions.
Different objectives have a variety of applications; see
\cite{bulucc2016recent} for a survey on the topic.
Connected balanced partitions in grid graphs has been of specific interest
with several NP-hardness results known, including
2-color weighted and only 3-rows \cite{becker1998max},
3-color unweighted \cite{berenger2018balanced}, and
$k$-color in solid grid graphs \cite{feldmann2013fast}.
On the positive side, Hettle et al.~\cite{hettle2021balanced} give
polynomial-time approximation algorithms for partitioning grid and
planar graphs, and apply these to real-world police and
fire-department districting problems.
We use only two colors in a solid unweighted grid, but introduce
the new constraint of pre-assigning more than one vertex to each partition.
When exactly one vertex of each color is already assigned,
the problem is often called a rooted partition problem.

\xxx{Coauthor also suggests it might be related to some of the districting/gerrymandering research if anyone is able to find that connection.}

\section{Hardness Proof}

All of the problems we consider are obviously in NP, by using the partition/solution as a witness.
We show that solving $n \times n$ Yin-Yang puzzles is NP-hard with or without
the $2 \times 2$ constraint by reductions
from the following NP-hard problem \cite{demaine2018tree}:

\begin{tcolorbox}
\emph{Planar 4-Regular Tree-Residue Vertex Breaking (TRVB)}:
Given a planar 4-regular multigraph, is there a subset of vertices that,
after being ``broken'', results in a single connected tree?
\emph{Breaking} a vertex involves deleting that vertex from the graph and adding a degree-1 vertex to each of its neighbors (thus modifying the edges incident to the broken vertex to instead be incident to a newly created degree-1 vertex):

\begin{center}
\includegraphics[width=0.8\linewidth]{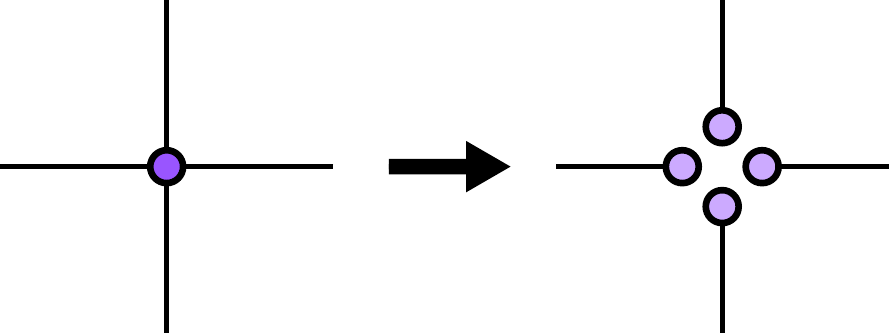}
\end{center}
\end{tcolorbox}

Figure~\ref{fig:TRVB_instance} shows an example of this problem.

\begin{figure}[hbt]
  \centering
  \includegraphics[width=\linewidth]{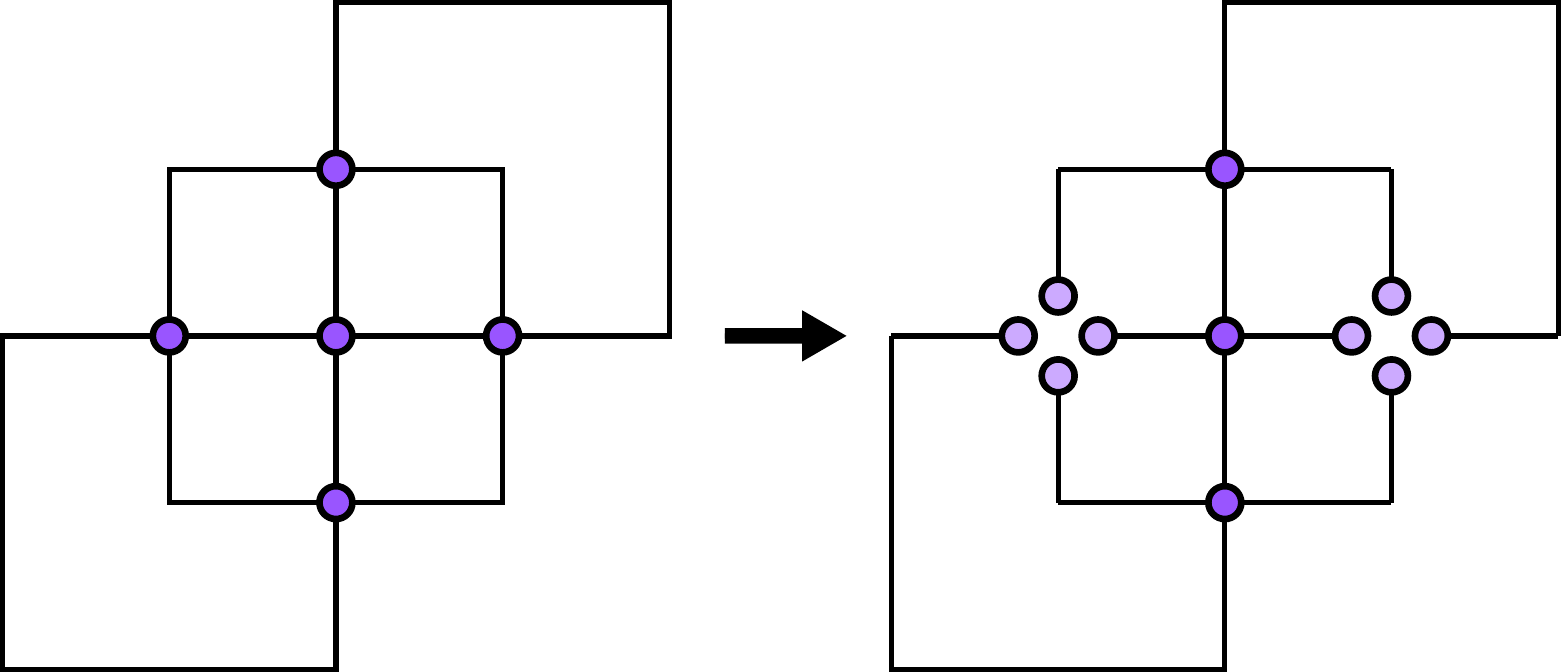}
  \caption{A 5-vertex instance of Planar 4-Regular TRVB (left)
    and a solution (right).}
  \label{fig:TRVB_instance}
\end{figure}



We will use the fact that planar maximum-degree-$4$ multigraphs can be drawn (with grid-routed edges) in a square grid of area $O(n^2)$ in polynomial time,
as in Figure~\ref{fig:TRVB_instance}.
We can use such an embedding for simple graphs (e.g., \cite{Papakostas-Tollis-1998}), and adapt it to multigraphs by subdividing multiple edges and loops, applying the simple embedding, and removing the subdivision vertices.

Before diving into the reductions, we develop a tool that helps force
local solutions to Yin-Yang puzzles:

\begin{lemma} \label{lem:diagnal}
	In a valid Yin-Yang puzzle solution (with or without the $2 \times 2$ constraint), no $2 \times 2$ square can contain diagonally opposite white and black circles.
\end{lemma}
\begin{proof}
	Consider the black circles. These two circles are not orthogonally adjacent, and thus must be connected by a path of black circles. However, this path will separate the pair of white circles which are also not orthogonally adjacent and thus must be connected by their own path of white circles.
\end{proof}

\subsection{Connected Partition: Without $2 \times 2$ Constraint}
\label{sec:nothin}

In this section, we prove NP-hardness of
Grid Graph Connected Partition Completion,
or equivalently, Yin-Yang puzzles without the $2 \times 2$ constraint.
We present this proof first as it is simpler but uses the same ideas.

In this reduction, edges will be represented by orthogonal paths
(\emph{wires}) of black circles,
and everything that is not an edge or vertex gadget
will be filled with white circles.

\begin{figure}
  \centering
  \subcaptionbox{\label{fig:nothin_TRVB_gadget} Vertex gadget}
    {\includegraphics{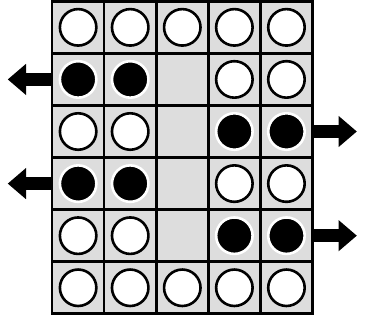}}

  \medskip

  \subcaptionbox{\label{fig:nothin_TRVB_gadget_broken} Broken solution}
    {\includegraphics{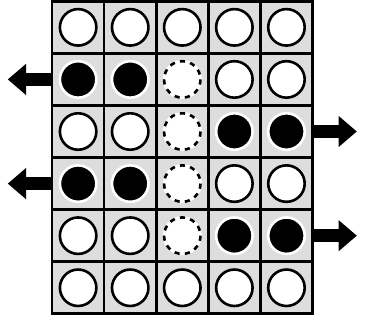}}\hfil
  \subcaptionbox{\label{fig:nothin_TRVB_gadget_unbroken} Unbroken solution}
    {\includegraphics{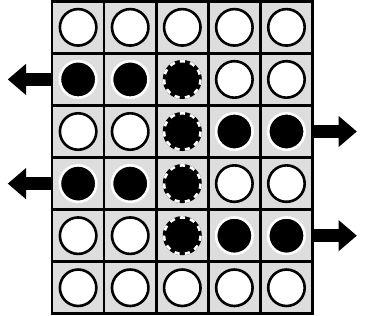}}
  \caption{Vertex gadget without the $2 \times 2$ constraint
    and its two possible local solutions.
    The arrows represent the four outgoing edges.}
  \label{fig:nothin_TRVB_gadget_all}
\end{figure}

\paragraph{Vertex Gadget.}
Figure~\ref{fig:nothin_TRVB_gadget} shows the vertex gadget.
If a solution fills the top empty cell (say) black or white,
then by repeated application of Lemma~\ref{lem:diagnal},
all four empty cells must be filled the same color.
In the local solution of Figure~\ref{fig:nothin_TRVB_gadget_broken},
the black-circle wires are disconnected from each other; while
in the local solution of Figure~\ref{fig:nothin_TRVB_gadget_unbroken},
the black-circle wires are all connected together.
Thus these two local solutions correspond to breaking and not breaking a
vertex of the TRVB instance.

\paragraph{Layout.}
We take an orthogonal grid drawing of the given TRVB graph,
and then scale the grid by a factor of 9.
This scaling allows us to place the $9 \times 9$ tiled versions of the
vertex gadget and edge gadgets (corners and straights) shown in
Figure~\ref{fig:nothin_routing},
which make it easy to align black-circle wires in row 3 and column 6.
All remaining cells are filled with white circles,
leaving empty only the four cells in each vertex gadget.

\begin{figure}
  \centering
  \subcaptionbox{Vertex gadget}{\includegraphics{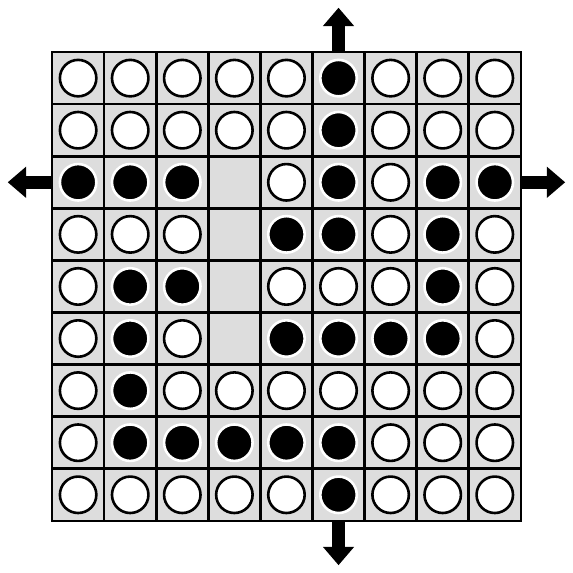}}

  \bigskip

  \subcaptionbox{One of six edge gadgets}{\includegraphics{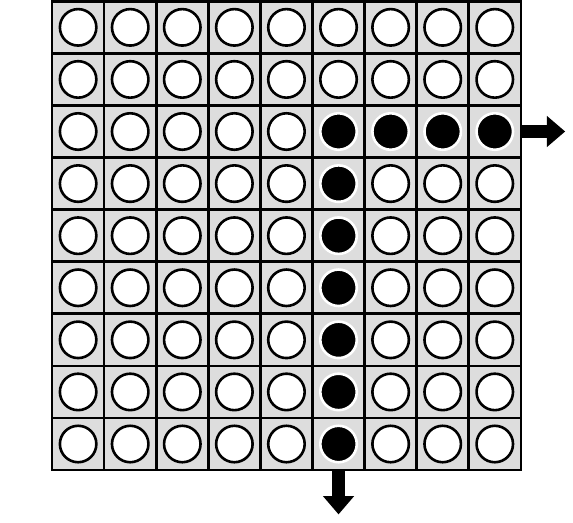}}
  \caption{$9 \times 9$ tiling versions of gadgets without the $2 \times 2$ constraint.}
  \label{fig:nothin_routing}
\end{figure}

\begin{theorem}
	It is NP-complete to decide whether there is a solution to a Yin-Yang puzzle
  without the $2 \times 2$ constraint,
  or Grid Graph Connected Partition Completion,
  on an $n \times n$ grid.
\end{theorem}
\begin{proof}
	If the TRVB instance has a solution, we fill broken vertices with four white circles (as in Figure~\ref{fig:nothin_TRVB_gadget_broken}) and unbroken vertices with four black circles (as in Figure~\ref{fig:nothin_TRVB_gadget_unbroken}). All edge gadgets touching an unbroken vertex are connected, and thus the fact that the TRVB solution is connected ensures the black circles form a single connected component. For a broken vertex, the added white circles connect the white circles in the four faces of the graph drawing. The TRVB solution being a tree ensures that this yields a single white connected component.

  Conversely, consider any solution to the Yin-Yang puzzle.
  As argued above, each vertex gadget must be solved using one of the two
  valid local solutions from Figure~\ref{fig:nothin_TRVB_gadget_all},
  and therefore we can translate this Yin-Yang solution to an assignment
  of whether to break each vertex in the TRVB instance.
  We claim that this assignment is in fact a solution to the TRVB instance.
  The region containing the black circles in the Yin-Yang solution
  has the same shape (in particular, toplogy) as the graph resulting from
  breaking the broken vertices in the candidate solution to the TRVB instance.
  If the resulting graph were disconnected,
  then the black circles would also form a disconnected region,
  contradicting the Yin-Yang connectivity constraint on black.
  If the resulting graph had a cycle,
  then the black circles would also form a nontrivial cycle,
  which would separate the white circles interior and exterior to that cycle,
  contradicting the Yin-Yang connectivity constraint on white.
  Therefore the resulting graph is connected and acyclic,
  so we have a solution to the TRVB instance.
  %
\end{proof}

\subsection{Tree Partition: With $2 \times 2$ Constraint}

In this section, we prove NP-hardness of 
Grid Graph Tree Partition Completion as well as Yin-Yang puzzles
(with both constraints).
We fully precolor the boundary vertices to not form a cycle,
so these problems become equivalent.
The reduction idea is the same as the proof without the $2 \times 2$ constraint from Section~\ref{sec:nothin}, but respecting this constraint requires a more complicated filler, and because of this filler, a more complex vertex gadget.
To get a sense of the overall structure, refer ahead to
Figure~\ref{fig:big_example} for a full example.

\subsubsection{Background Filler}

The background filler consists of alternating columns of white and black circles, except for a full row of black circles at the top and a full row of white circles at the bottom (to ensure connectivity); see Figure~\ref{fig:filler}.
On top of this filler gadget, we draw vertex and edge gadgets.

\begin{figure}
  \centering
  \includegraphics{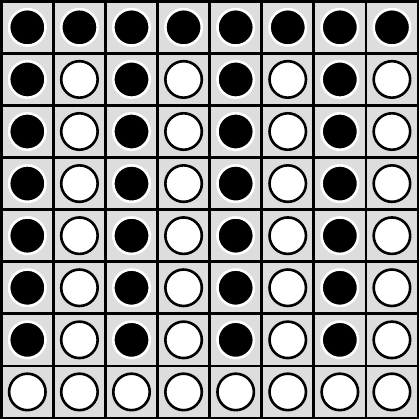}
  \caption{Background filler satisfying the $2 \times 2$ constraint.}
  \label{fig:filler}
\end{figure}

\subsubsection{Edge Gadget}

We again represent edges as orthogonal paths (wires) of black circles but now, wherever an edge travels horizontally, we also add a row of white circles immediately above the path, except where the edge turns upward; see Figure~\ref{fig:edge gadget}.
Because of the background filler, each horizontal segment of an edge gadget
will have black downward tendrils (paths) from every other column,
similar to the downward black tendrils from the top row of the filler.
(Similarly, the white circles immediately above a horizontal segment of an
edge gadget will have white upward tendrils from every other column,
similar to the upward white tendrils from the bottom row of the filler.)
Edge gadgets can travel vertically only on the black parity of the background
filler, so that these circles of the edge gadget match the background filler.

\begin{figure}
  \centering
  \raisebox{0.3in}{\includegraphics[scale=0.66]{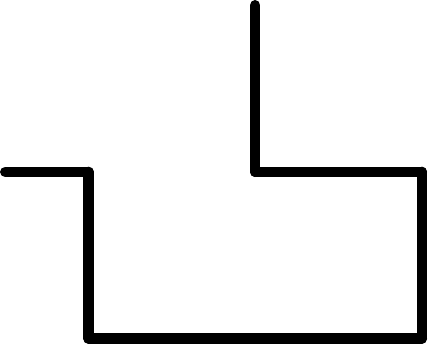}}\hfil
  \includegraphics{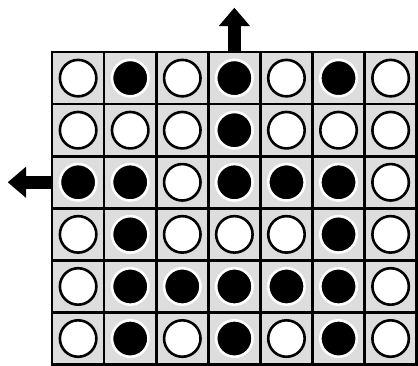}
  \caption{An edge route (left) and the corresponding edge gadget (right).}
  \label{fig:edge gadget}
\end{figure}

\subsubsection{Vertex Gadget}

Figure~\ref{fig:TRVB_vertex} shows the vertex gadget.
By the $2 \times 2$ constraint, the leftmost two empty cells must be filled
with circles of opposite color; refer to Figure~\ref{fig:TRVB_vertex_solutions}.
Once this choice gets made, repeated application of Lemma~\ref{lem:diagnal}
forces the coloring of the horizontal rows of empty squares to be uniform
and opposite from each other, and then forces the rightmost two empty cells
to match the leftmost two empty cells.
The top three empty cells are then forced to be filled in with circles of
the same color as the bottom row, in order to prevent local isolation of
black or white circles.

\begin{figure}[ht]
    \centering
    \includegraphics{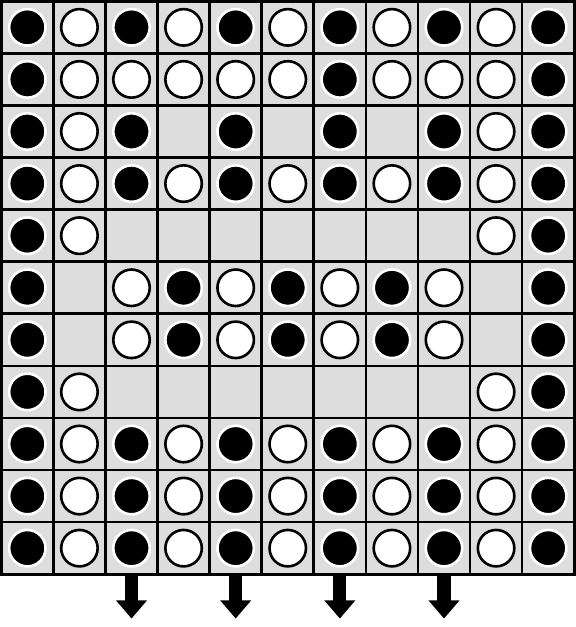}
    \caption{Vertex gadget.  The arrows represent the four outgoing edges.}
    \label{fig:TRVB_vertex}
\end{figure}

Figure~\ref{fig:TRVB_vertex_solutions}
shows the two possible resulting local solutions.
Figure~\ref{fig:TRVB_vertex_solution_broken}
corresponds to breaking the TRVB vertex,
as it keeps the four incident edges disconnected from each other;
while Figure~\ref{fig:TRVB_vertex_solution_unbroken}
corresponds to not breaking the corresponding TRVB vertex,
as it connects together the four incident edges.
Both solutions also connect together the five white paths at the top and
the two outermost white paths on the bottom, and the broken solution further
connects these white paths to the three middle white paths at the bottom.
These connections correspond to exactly one white region per face around the
vertex: one face for a broken vertex and four faces for an unbroken vertex.

\begin{figure}
  \begin{subfigure}{\columnwidth}
    \centering
    \includegraphics{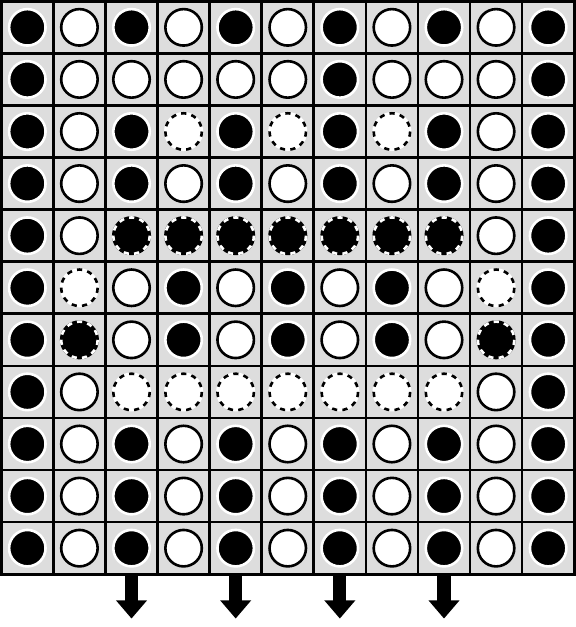}
    \caption{Broken vertex solution.}
    \label{fig:TRVB_vertex_solution_broken}
  \end{subfigure}

  \bigskip

  \begin{subfigure}{\columnwidth}
    \centering
    \includegraphics{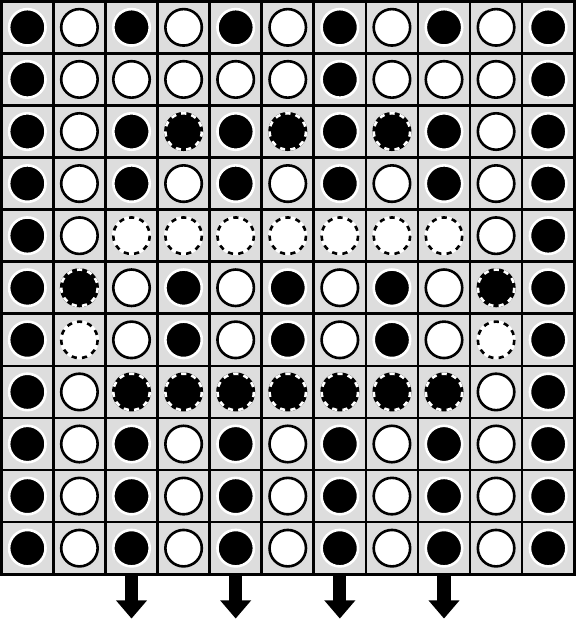}
    \caption{Unbroken vertex solution.}
    \label{fig:TRVB_vertex_solution_unbroken}
  \end{subfigure}
  \caption{The two local solutions to the vertex gadget of
    Figure~\ref{fig:TRVB_vertex}.}
  \label{fig:TRVB_vertex_solutions}
\end{figure}

\subsubsection{Reduction Layout}
\label{sec:layout}

We take an orthogonal grid drawing of the given TRVB graph,
and then scale the grid by a factor of 16.
This scaling allows us to place the $16 \times 16$ tiled versions of the
vertex gadget and edge gadgets (corners and straights)
shown in Figure~\ref{fig:16x16_routing},
which make it easy to align black-circle wires in row 12 and column 11.
Any absent $16 \times 16$ squares get filled with alternating columns of
black and white.
Then we add an extra row at the top and bottom of the puzzle,
filled with black and white circles respectively, to complete the
filler of Figure~\ref{fig:filler}.
Finally, for one topmost horizontal segment of an edge gadget,
in a black-parity column, we change one of the cells above the segment
from white to black, thereby attaching the edge gadget to a black tendril
and thus the top row of black circles;
we call the changed cell the \emph{exceptional cell}.

\begin{figure}
  \centering
  \subcaptionbox{Vertex gadget}
    {\includegraphics[width=\columnwidth]{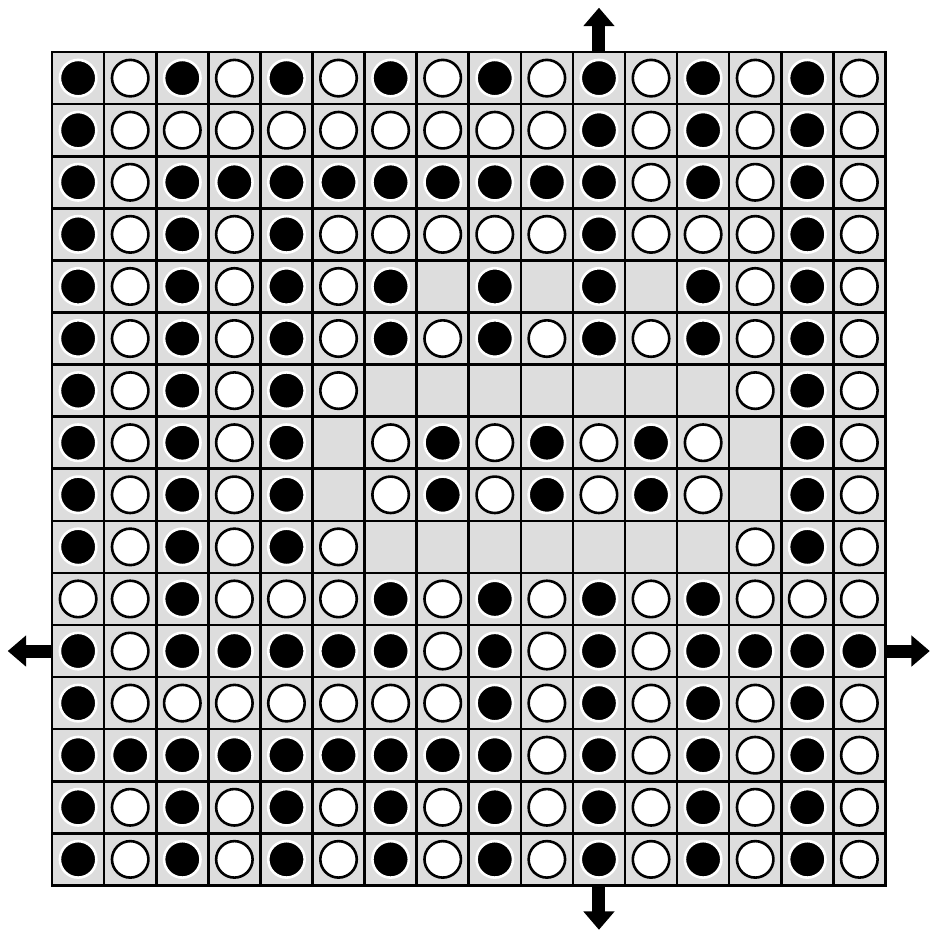}}

  \bigskip

  \subcaptionbox{One of six edge gadgets}
    {\includegraphics[width=\columnwidth]{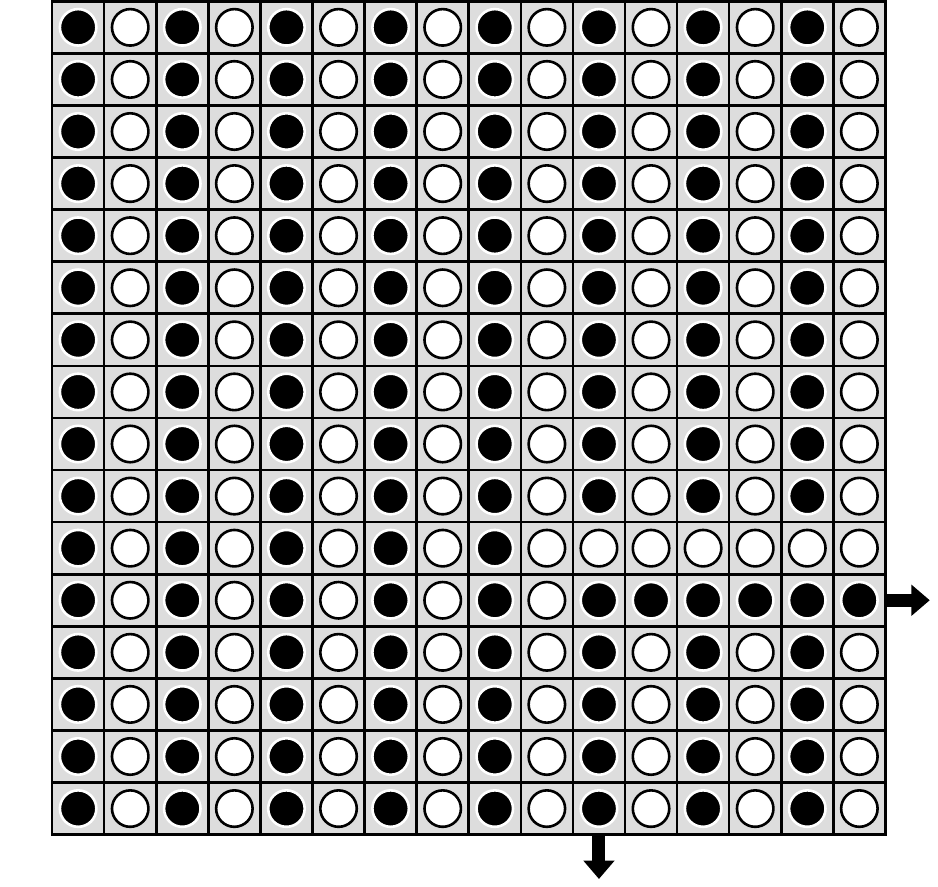}}
  \caption{$16 \times 16$ tiling versions of the vertex and edge gadgets
    from Figures~\ref{fig:TRVB_vertex} and~\ref{fig:edge gadget}.}
  \label{fig:16x16_routing}
\end{figure}

Figure~\ref{fig:big_example} shows a complete example of the reduction
applied to the instance from Figure~\ref{fig:TRVB_instance},
and Figure~\ref{fig:big_example_solved} shows the corresponding solution.

\begin{figure*}
	\centering

  \gdef\overlayfig#1{%
    \vspace{0.0125\linewidth} 

    \begin{overpic}[width=\linewidth]{#1}
      \color{purple}
      \linethickness{2pt}
      \put(0,100){\includegraphics[width=\linewidth]{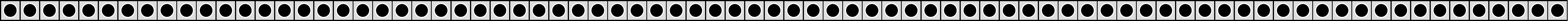}}
      \put(0,-1.25){\includegraphics[width=\linewidth]{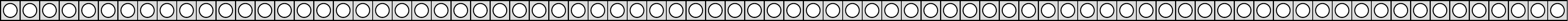}}
      \def\eps{0.175}
      \def\twoeps{.350}
      \put(-\eps,  0){\line(1,0){100\twoeps}}
      \put(-\eps, 20){\line(1,0){100\twoeps}}
      \put(-\eps, 40){\line(1,0){100\twoeps}}
      \put(-\eps, 60){\line(1,0){100\twoeps}}
      \put(-\eps, 80){\line(1,0){100\twoeps}}
      \put(-\eps,100){\line(1,0){100\twoeps}}
      \put(  0,-\eps){\line(0,1){100\twoeps}}
      \put( 20,-\eps){\line(0,1){100\twoeps}}
      \put( 40,-\eps){\line(0,1){100\twoeps}}
      \put( 60,-\eps){\line(0,1){100\twoeps}}
      \put( 80,-\eps){\line(0,1){100\twoeps}}
      \put(100,-\eps){\line(0,1){100\twoeps}}
    \end{overpic}

    \vspace{0.0125\linewidth} 
  }%
  \overlayfig{big_example_16x16}

	\caption{Full reduction from the 5-vertex instance of TRVB from
    Figure~\ref{fig:TRVB_instance}.}
	\label{fig:big_example}
\end{figure*}

\begin{figure*}
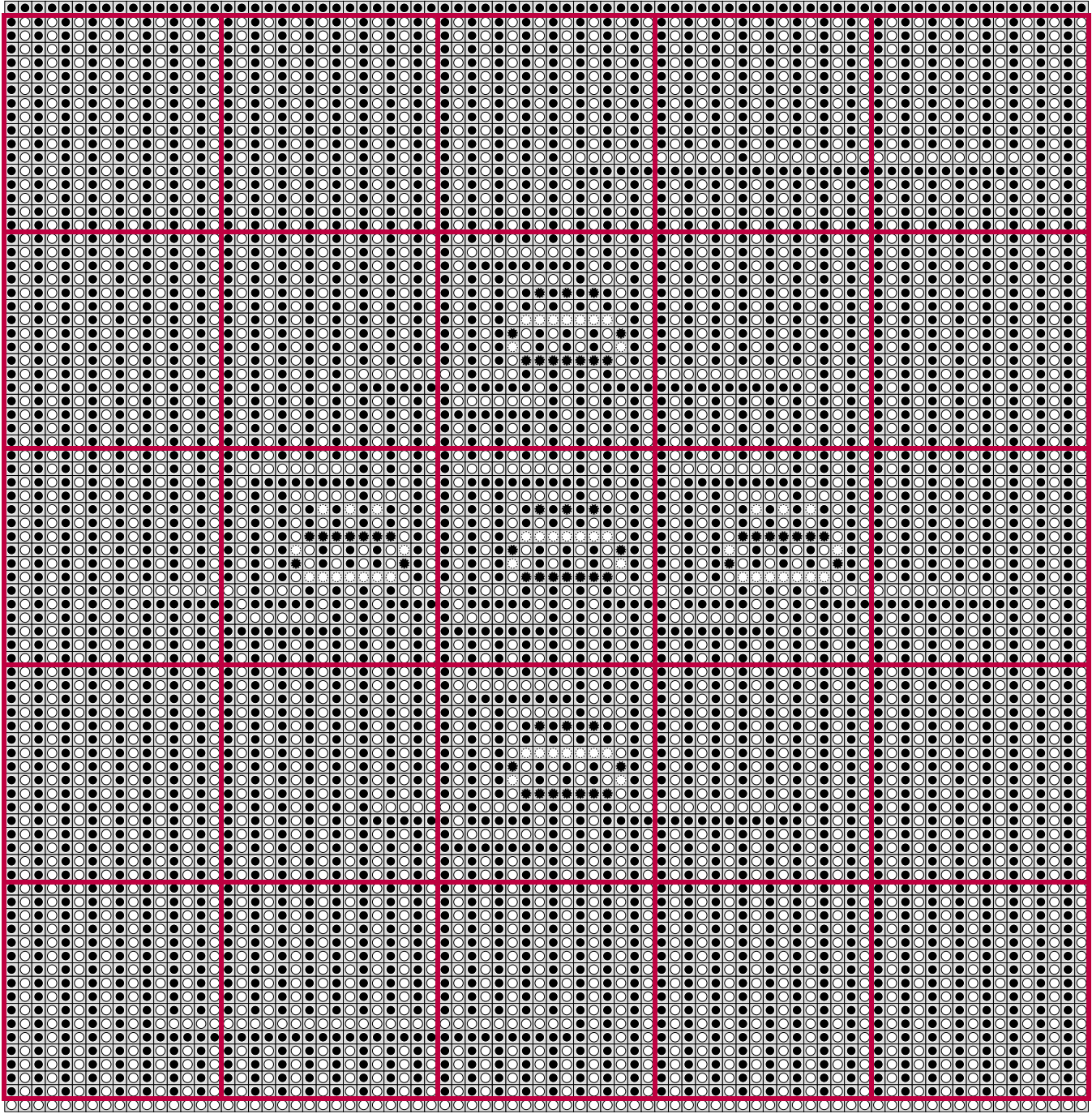

	\centering

  \overlayfig{big_example_16x16_solved}

	\caption{Solution to the puzzle in Figure~\ref{fig:big_example}
    corresponding to the TRVB solution in Figure~\ref{fig:TRVB_instance}.}
	\label{fig:big_example_solved}
\end{figure*}




\begin{theorem}
  It is NP-complete to decide whether there is a solution to a Yin-Yang puzzle
  (with both constraints),
  or Grid Graph Tree Partition Completion,
  on an $n \times n$ grid.
\end{theorem}

\begin{proof}
  Because our reduction instance has black and white circles on the boundary,
  the Yin-Yang puzzle with $2 \times 2$ constraint
  is equivalent to Grid Graph Tree Partition Completion.
  Furthermore, if the black circles form a connected and acyclic
  subset of cells, then so do the white circles: if the white circles formed
  an induced cycle of length $> 4$, then there would be black circles
  both interior (because the cycle is induced) and exterior
  (on the boundary), a contradiction.
  Therefore, it suffices to prove that the black circles are
  connected and acyclic if and only if the chosen broken/unbroken solutions 
  from Figure~\ref{fig:TRVB_vertex_solutions} for each vertex gadget
  corresponds to a solution to the TRVB instance.

  Define the \emph{important} cells to consist of the following:
  \begin{enumerate}
  \item black edge wires, or more precisely, the shortest paths
    among black circles connecting pairs of the vertex gadgets' ports
    (marked with arrows in Figure~\ref{fig:TRVB_vertex}); and
  \item the bottom row of initially empty cells in each vertex gadget (instance
    of Figure~\ref{fig:TRVB_vertex}), together with the shortest paths of black
    circles within the gadget connecting those cells to the ports (arrows).
  \end{enumerate}
  The important cells directly represent the connectivity of the TRVB instance,
  so the important black circles are connected and acyclic if and only if the
  Yin-Yang solution corresponds to a TRVB solution.

  Thus it suffices to show that the important black circles are connected and
  acyclic if and only if all of the black circles are connected and acyclic.
  Unimportant black circles are either part of a vertex gadget,
  part of the top row, or part of a downward black tendril
  from the top row or a horizontal segment of an edge gadget.
  Figure~\ref{fig:TRVB_vertex_solutions} shows that unimportant black circles
  within a vertex gadget are all attached to the downward black tendrils above
  the gadget, with one acyclic connected component per tendril,
  aside from a few black circles (in the unbroken case)
  attached directly to important black circles.
  For unimportant black circles in black tendrils, we can trace the
  tendrils up to find their connection to either important black circles
  or the top row.
  (Tendrils are never connected to black circles at their bottom,
  except to $O(1)$ black circles within a vertex gadget.)
  The black tendrils that drop from an important black circle are all connected
  together if and only if the important black circles are themselves connected.
  The black tendrils that drop from the top row are all connected together
  via the top row of unimportant black circles, and the exceptional cell
  connects one of these tendrils to the important black circles.
  Therefore the unimportant black circles form trees attached to the
  important black cells, so they do not affect connectivity or acyclicity.
\end{proof}

\section{Open Problems}
Our reduction from Planar 4-regular TRVB is parsimonious (preserves the number of solutions).
It is unknown whether the Another Solution Problem or counting versions of TRVB are (NP- or \#P-) hard,
but such results would carry over to Yin-Yang puzzles.

In our reductions, we fill almost the entire board with circles. However, elegant puzzles tend to have only a few pre-marked circles and so one may wonder whether this problem remains hard with a much sparser clue set. It seems likely small modifications to our reduction would allow for reductions with only $O(n)$ circles, but getting $o(n)$ would require fundamentally different techniques and some sort of clever encoding of problems in the relative spacing of the circles. On the extreme, it seems reasonably likely this problem is FPT with respect to the number of pre-assigned circles.

Both from a computational and puzzle design standpoint, we could imagine generalizing the number of colors of circles to more than two, or looking at other families of graphs. The $2 \times 2$ constraint can generalize to forbidding all vertices around a face from being the same color in an embedded graph. It seems very natural to look at such puzzles on triangular and hexagonal grids and in those cases the three-fold symmetry may aesthetically adapt to three colors. Generalizing the number of colors trivially remains NP-hard, as we can simply add a row for each additional color at the top of our two color reduction, ensuring those new colors will not be able to be placed. We also believe similar proof techniques will work for hexagonal and triangular grids.

\section*{Acknowledgments}

This research started at the 34th Bellairs Winter Workshop on Computational
Geometry, co-organized by Erik Demaine and Godfried Toussaint in 2019.
We thank the other workshop participants for fruitful discussions
and for providing an inspiring research atmosphere. 

Research supported in part by NSERC and JSPS KAKENHI Grant Numbers
JP17K00017, JP21K11757 and JP20H05964.

Most figures of this paper were drawn using SVG Tiler
[\url{https://github.com/edemaine/svgtiler}].
Source files for these figures are freely available
[\url{https://github.com/edemaine/yin-yang-svgtiler}].

\bibliographystyle{plain}
\bibliography{yinyang}

\begin{thebibliography}{10}

\bibitem{Witness_FUN2018}
Zachary Abel, Jeffrey Bosboom, Erik~D. Demaine, Linus Hamilton, Adam
  Hesterberg, Justin Kopinsky, Jayson Lynch, and Mikhail Rudoy.
\newblock Who witnesses {T}he {W}itness? {F}inding witnesses in {T}he {W}itness
  is hard and sometimes impossible.
\newblock In {\em Proceedings of the 9th International Conference on Fun with
  Algorithms (FUN 2018)}, pages 3:1--3:21, La Maddalena, Italy, June 2018.

\bibitem{NumberlessShakashaka_CCCG2015}
Aviv Adler, Michael Biro, Erik Demaine, Mikhail Rudoy, and Christiane Schmidt.
\newblock Computational complexity of numberless {S}hakashaka.
\newblock In {\em Proceedings of the 27th Canadian Conference on Computational
  Geometry (CCCG 2015)}, Kingston, Canada, August 2015.

\bibitem{adler2020tatamibari}
Aviv Adler, Jeffrey Bosboom, Erik~D. Demaine, Martin~L. Demaine, Quanquan~C.
  Liu, and Jayson Lynch.
\newblock Tatamibari is {NP}-complete.
\newblock In {\em Proceedings of the 10th International Conference on Fun with
  Algorithms (FUN 2020)}, LIPIcs, 2020.

\bibitem{allen2018sto}
Addison Allen and Aaron Williams.
\newblock Sto-stone is {NP}-complete.
\newblock In {\em Proceedings of the 30th Canadian Conference on Computational
  Geometry (CCCG 2018)}, pages 28--34, 2018.

\bibitem{Andersson-2007}
Daniel Andersson.
\newblock {HIROIMONO} is {NP}-complete.
\newblock In {\em Proceedings of the 4th International Conference on Fun with
  Algorithms (FUN 2007)}, volume 4475 of {\em Lecture Notes in Computer
  Science}, pages 30--39, 2007.

\bibitem{Hashiwokakero}
Daniel Andersson.
\newblock Hashiwokakero is {NP}-complete.
\newblock {\em Information Processing Letters}, 109(19):1145--1146, 2009.

\bibitem{arkin2009not}
Esther~M. Arkin, S{\'a}ndor~P. Fekete, Kamrul Islam, Henk Meijer, Joseph S.~B.
  Mitchell, Yurai N{\'u}{\~n}ez-Rodr{\'\i}guez, Valentin Polishchuk, David
  Rappaport, and Henry Xiao.
\newblock Not being (super) thin or solid is hard: A study of grid
  hamiltonicity.
\newblock {\em Computational Geometry: Theory and Applications},
  42(6--7):582--605, 2009.

\bibitem{becker1998max}
Ronald Becker, Isabella Lari, Mario Lucertini, and Bruno Simeone.
\newblock Max-min partitioning of grid graphs into connected components.
\newblock {\em Networks}, 32(2):115--125, 1998.

\bibitem{berenger2018balanced}
Cedric Berenger, Peter Niebert, and Kevin Perrot.
\newblock Balanced connected partitioning of unweighted grid graphs.
\newblock In {\em Proceedings of the 43rd International Symposium on
  Mathematical Foundations of Computer Science (MFCS 2018)}, volume 117 of {\em
  LIPIcs}, 2018.

\bibitem{bulucc2016recent}
Ayd{\i}n Bulu{\c{c}}, Henning Meyerhenke, Ilya Safro, Peter Sanders, and
  Christian Schulz.
\newblock Recent advances in graph partitioning.
\newblock In {\em Selected Results and Surveys on Algorithm Engineering},
  volume 9220 of {\em Lecture Notes in Computer Science}, pages 117--158.
  Springer, 2016.

\bibitem{yin-yang-original}
{\begin{CJK}{UTF8}{ipxm}世界文化社\end{CJK}}.
\newblock {\begin{CJK}{UTF8}{ipxm}しろまるくろまる\end{CJK}}.
\newblock {\em {\begin{CJK}{UTF8}{ipxm}パズラー \end{CJK} (Puzzler)}}, 150,
  May 1994.

\bibitem{Demaine-Hearn-2009-survey-both}
Erik~D. Demaine and Robert~A. Hearn.
\newblock Playing games with algorithms: Algorithmic combinatorial game theory.
\newblock In {\em Games of No Chance 3}, pages 3--56. Cambridge University
  Press, 2009.
\newblock arXiv:cc.CC/0106019.

\bibitem{Shakashaka}
Erik~D. Demaine, Yoshio Okamoto, Ryuhei Uehara, and Yushi Uno.
\newblock Computational complexity and an integer programming model of
  {S}hakashaka.
\newblock {\em IEICE Transactions on Fundamentals of Electronics,
  Communications and Computer Sciences}, E97-A(6):1213--1219, 2014.

\bibitem{demaine2017thin}
Erik~D. Demaine and Mikhail Rudoy.
\newblock Hamiltonicity is hard in thin or polygonal grid graphs, but easy in
  thin polygonal grid graphs.
\newblock arXiv:1706.10046, 2017.
\newblock \url{https://arXiv.org/abs/1706.10046}.

\bibitem{demaine2018tree}
Erik~D. Demaine and Mikhail Rudoy.
\newblock {Tree-Residue Vertex-Breaking}: a new tool for proving hardness.
\newblock In {\em Proceedings of the 16th Scandinavian Symposium and Workshops
  on Algorithm Theory (SWAT 2018)}, volume 101 of {\em LIPIcs}, pages
  32:1--32:14, 2018.

\bibitem{feldmann2013fast}
Andreas~Emil Feldmann.
\newblock Fast balanced partitioning is hard even on grids and trees.
\newblock {\em Theoretical Computer Science}, 485:61--68, 2013.

\bibitem{Friedman-2002-corral}
Erich Friedman.
\newblock Corral puzzles are {NP}-complete.
\newblock \url{https://erich-friedman.github.io/papers/corral.pdf}, August
  2002.

\bibitem{Friedman-2002-pearl}
Erich Friedman.
\newblock Pearl puzzles are {NP}-complete.
\newblock \url{https://erich-friedman.github.io/papers/pearl.pdf}, August 2002.

\bibitem{SpiralGalaxies}
Erich Friedman.
\newblock Spiral {Galaxies} puzzles are {NP-complete}.
\newblock \url{https://erich-friedman.github.io/papers/spiral.pdf}, March 2002.

\bibitem{GPC}
Robert~A. Hearn and Erik~D. Demaine.
\newblock {\em Games, Puzzles, and Computation}.
\newblock A K Peters, July 2009.

\bibitem{hettle2021balanced}
Cyrus Hettle, Shixiang Zhu, Swati Gupta, and Yao Xie.
\newblock Balanced districting on grid graphs with provable compactness and
  contiguity.
\newblock arXiv:2102.05028, 2021.
\newblock \url{https://arXiv.org/abs/2102.05028}.

\bibitem{Holzer-Klein-Kutrib-2004}
Markus Holzer, Andreas Klein, and Martin Kutrib.
\newblock On the {NP}-completeness of the {\textsc{nurikabe}} pencil puzzle and
  variants thereof.
\newblock In {\em Proceedings of the 3rd International Conference on Fun with
  Algorithms (FUN 2004)}, pages 77--89, Isola d'Elba, Italy, May 2004.

\bibitem{Holzer-Ruepp-2007}
Markus Holzer and Oliver Ruepp.
\newblock The troubles of interior design---a complexity analysis of the game
  {H}eyawake.
\newblock In {\em Proceedings of the 4th International Conference on Fun with
  Algorithms (FUN 2007)}, volume 4475 of {\em Lecture Notes in Computer
  Science}, pages 198--212, 2007.

\bibitem{YajilinandCountryRoad}
Ayaka Ishibashi, Yuichi Sato, and Shigeki Iwata.
\newblock {NP}-completeness of two pencil puzzles: {Y}ajilin and {C}ountry
  {R}oad.
\newblock {\em Utilitas Mathematica}, 88:237--246, 2012.

\bibitem{Yosenabe}
Chuzo Iwamoto.
\newblock Yosenabe is {NP}-complete.
\newblock {\em Journal of Information Processing}, 22(1):40--43, 2014.

\bibitem{iwamoto2018computational}
Chuzo Iwamoto and Masato Haruishi.
\newblock Computational complexity of usowan puzzles.
\newblock {\em IEICE Transactions on Fundamentals of Electronics,
  Communications and Computer Sciences}, 101(9):1537--1540, 2018.

\bibitem{iwamoto2018dosun}
Chuzo Iwamoto and Tatsuaki Ibusuki.
\newblock Dosun-fuwari is {NP}-complete.
\newblock {\em Journal of Information Processing}, 26:358--361, 2018.

\bibitem{iwamoto2020polynomial}
Chuzo Iwamoto and Tatsuaki Ibusuki.
\newblock Polynomial-time reductions from {3SAT} to kurotto and juosan puzzles.
\newblock {\em IEICE Transactions on Information and Systems}, 103(3):500--505,
  2020.

\bibitem{kendall2008survey}
Graham Kendall, Andrew Parkes, and Kristian Spoerer.
\newblock A survey of {NP}-complete puzzles.
\newblock {\em ICGA Journal}, 31(1):13--34, 2008.

\bibitem{Kurodoko}
Jonas K{\"o}lker.
\newblock Kurodoko is {NP}-complete.
\newblock {\em Journal of Information Processing}, 20(3):694--706, 2012.

\bibitem{NumberlinkNP}
Kouichi Kotsuma and Yasuhiko Takenaga.
\newblock {NP}-completeness and enumeration of {N}umber {L}ink puzzle.
\newblock {\em IEICE Technical Report}, 109(465):1--7, March 2010.

\bibitem{McPhail-2003}
Brandon McPhail.
\newblock The complexity of puzzles.
\newblock Undergraduate thesis, Reed College, Portland, Oregon, 2003.

\bibitem{McPhail-2005-LightUp}
Brandon McPhail.
\newblock Light {U}p is {NP}-complete.
\newblock
  \url{http://www.mountainvistasoft.com/docs/lightup-is-np-complete.pdf}, 2005.

\bibitem{McPhail-2007-LITS}
Brandon McPhail.
\newblock Metapuzzles: Reducing {SAT} to your favorite puzzle.
\newblock CS Theory talk, December 2007.

\bibitem{yin-yang-example}
Gareth Moore.
\newblock {\em Paper, Pencil \& You: Mindfulness: Relaxing Brain-Training
  Puzzles for Stressed-Out People}.
\newblock Greenfinch, 2020.

\bibitem{is2018paper}
Daniel Packer, Sophia White, and Aaron Williams.
\newblock A paper on {P}encils: A pencil and paper puzzle: {P}encils is
  {NP}-complete, pencils.
\newblock In {\em Proceedings of the 30th Canadian Conference on Computational
  Geometry (CCCG 2018)}, page~35, 2018.

\bibitem{Papakostas-Tollis-1998}
Achilleas Papakostas and Ioannis~G. Tollis.
\newblock Algorithms for area-efficient orthogonal drawings.
\newblock {\em Computational Geometry: Theory and Applications}, 9(1):83--110,
  1998.

\bibitem{Ueda-Nagao-1996}
Nobuhisa Ueda and Tadaaki Nagao.
\newblock {NP}-completeness results for {NONOGRAM} via parsimonious reductions.
\newblock Technical Report TR96-0008, Department of Computer Science, Tokyo
  Institute of Technology, Tokyo, Japan, May 1996.

\bibitem{uejima2015fillmat}
Akihiro Uejima and Hiroaki Suzuki.
\newblock Fillmat is {NP}-complete and {ASP}-complete.
\newblock {\em Journal of Information Processing}, 23(3):310--316, 2015.

\bibitem{Yato-2003}
Takayuki Yato.
\newblock Complexity and completeness of finding another solution and its
  application to puzzles.
\newblock Master's thesis, University of Tokyo, Tokyo, Japan, January 2003.

\bibitem{Yato-Seta-2003}
Takayuki Yato and Takahiro Seta.
\newblock Complexity and completeness of finding another solution and its
  application to puzzles.
\newblock {\em IEICE Transactions on Fundamentals of Electronics,
  Communications, and Computer Sciences}, E86-A(5):1052--1060, 2003.
\newblock Also IPSJ SIG Notes 2002-AL-87-2, 2002.

\end{thebibliography}


\end{document}